\documentclass[11pt,fleqn]{llncs}

\usepackage{syntonly}
\usepackage{latexsym}
\usepackage{amssymb}
\usepackage{epsfig}

\renewcommand{\H}{\mbox{$\bf{H}$}}
\newcommand{\G}{\mbox{$\bf{G}$}}
\newcommand{\ord}{\mbox{$o$}}
\newcommand{\Ord}{\mbox{$O$}}
\newcommand{\DT}{\mbox{\rm DTIME}}
\newcommand{\NT}{\mbox{\rm NTIME}}
\newcommand{\AT}{\mbox{\rm ATIME}}
\newcommand{\XT}{\mbox{\rm XTIME}}

\newcommand{\NS}{\mbox{\rm NSPACE}}

\newcommand{\redlin}{\mbox{$\leq_{\mbox{\rm\footnotesize lin}}$}}
\newcommand{\redlino}{\mbox{$\leq_{\mbox{\rm\footnotesize lin,1}}$}}

\newcommand{\code}{\mbox{\rm code}}

\newcommand{\NP}{\mbox{\rm NP}}
\renewcommand\P{\mbox{\rm P}}

\newcommand{\co}{\mbox{\rm co}}

\newcommand{\N}{\mbox{$\mathbb{N}$}}                        
\newcommand{\Ni}{\mbox{\footnotesize $\mathbb{N}$}}         
\newcommand{\Nplus}{\mbox{$\mathbb{N}_+\,$}}

\newcommand{\fct}[3]{\mbox{$#1 :\, #2\, \longrightarrow \, #3$}}

\newcommand{\en}{\enspace}

\newcommand{\bea}{\begin{eqnarray*}}
\newcommand{\eea}{\end{eqnarray*}}

\begin{document}

 \setcounter{page}{1}

\thispagestyle{empty}

\vspace*{0.3cm}

\begin{center}\Large\bf On the Tape-Number Problem \\
 for Deterministic Time Classes
  \vspace*{0.5cm}

\normalsize\rm Armin Hemmerling
 \vspace*{0.12cm}

\small Ernst-Moritz-Arndt--Universit\"at Greifswald,
           Institut f\"ur Mathematik und Informatik\\
 Walther-Rathenau-Str. 47,
           D--17487 Greifswald, Germany\\
{\small\tt hemmerli@uni-greifswald.de}
 \vspace*{0.12cm}

\normalsize
  September 2012
\end{center}
 \vspace*{1.cm}

\hspace*{0.25cm}
\begin{minipage}{14.cm}
\small
\noindent
 {\bf Abstract.} \en
For any time bound $f$, let $\H(f)$ denote the hierarchy conjecture which means that the restriction of the numbers of work tapes of deterministic Turing machines to $b\in\N$ generates an infinite hierarchy of proper subclasses $\DT_b(f) \subset \DT(f)$. We show that $\H(f)$ implies separations of deterministic from nondeterministic time classes.
$\H(f)$ follows from the gap property, $\G(f)$, which says that there is a time-constructible $f_2$ such that $f\in\ord(f_2)$ and $\DT(f)=\DT(f_2)$. $\G(f)$ implies further separations.
All these relationships relativize.
   \\[1ex]
{\bf MSC (2010):} \en   68Q15, 03D15 \\[0.7ex]
{\bf Keywords:} \en  Turing machine, time complexity, tape number, determinism versus nondeterminism, gap property, relativization.
  \end{minipage}
\vspace*{0.8cm}

\section{Introduction and overview}

For nondeterministic and alternating Turing machines, the number of work tapes can be restricted to two and even to one, respectively, without any loss of time, see facts $(1)$ and $(2)$ in the next section. This implies rather strong hierarchy results for the related time complexity classes. For deterministic acceptors the situation is different. The best known technique to reduce the number of work tapes to some constant comes from the early days of complexity theory. It causes a slow down from a time bound $f(n)$ to $f(n)\cdot \log(\,f(n)\,)$, see fact $(3)$. As a consequence, the deterministic time hierarchy theorem states the separation of the deterministic time classes, $\DT(f_1)\subset\DT(f_2)$, only under the supposition that
$f_1\cdot(\log\circ f_1)\in \ord(f_2)$ and $f_2$ is time-constructible.

In \cite{CN} it is considered plausible that $(b+1)$-tape deterministic linear-time Turing machines are strictly more powerful than $b$-tape such machines. Using the notation recalled in the next section, this means the conjecture that $\DT_b(n) \subset \DT_{b+1}(n)$. It holds for $b=0$. By \cite{Hen}, $\DT_1(n)$ even contains languages that cannot be accepted by any single-tape off-line machine (without a separate input tape) working in linear time, e.g., the language of palindromes. Recently \cite{TYL} it was shown that nondeterministic single-tape $\ord(n\cdot \log(n))$-time Turing machines accept only regular languages. By an involved analysis of computation graphs, a proof of the strictness of the inclusion $\DT_1(n) \subset \DT_2(n)$ was given in \cite{MSST}. For $b\geq 2$, however, no proof of the above conjecture has been presented so far.

To complete the picture, we mention some results concerning more restrictive models of computation. The analogous strict hierarchy property with respect to the numbers of tapes holds for deterministic real-time Turing machines by \cite{Aa}. More generally, $b$-tape Turing machines require more than linear time in order to on-line simulate $(b+1)$-tape machines,
see \cite{Pa}. For a unifying view and a discussion of these results, the reader is referred to \cite{PSSN}.

In this paper, we consider a weaker version of the above conjecture and generalize it to arbitrary time bounds $f$.
This yields the following definition.
\begin{definition}[Hierarchy Property]
For a time bound $f$, let $\,\H(f)\,$  mean that
$\, \DT_b(f)\not=\DT(f)\,$  for all $b\in\N$.
\end{definition}

We shall show that, for a variety of functions $f$, $\H(f)$ implies the separation $\DT(f)\not=\NT(f)$.
Moreover, $\H(f)$ is downward hereditary, see Lemma 3.
Thus, a proof of $\H(f)$, if it is possible at all, would have to be rather involved and can hardly be expected here, in particular for superlinear time bounds $f$.

 Another property we shall deal with, too, means that there is a gap of growth of complexity classes defined by time-constructible bounds just above the function $f$ (for the notion of time-constructibility, we refer to the next section):
\begin{definition}[Gap Property]
For a time bound $f$, let $\,\G(f)\,$  mean that
there is a time-constructible time bound $f_2$ such that
 $f\in\ord(f_2)$  and  $\DT(f)=\DT(f_2)$.
\end{definition}

Even if it has not yet been expressed in this form, the gap property is obviously closely related to the well-known problem whether or not the deterministic time hierarchy can be refined. More precisely, the negation, $\neg\,\G(f)$, would just mean a stronger hierarchy result for the complexity classes of time bounds immediately above $f$.

 It will turn out that $\G(f)$ implies both $\H(f)$ as well as, in many cases, $\DT(f_2)\not=\NT(f_2)$ for the related bounds $f_2$.
 An upward transfer of $\G(f)$ holds by Lemma 5.
 Thus, we cannot expect to prove $\G(f)$ by means of the known standard techniques.
So, on the one hand, our results might point out a new way towards separations of determinism from nondeterminism. On the other hand, they demonstrate the hardness of confirming properties like $\H(f)$ or $\G(f)$.

The paper is organized as follows.
Section 2 explains the basic notions and provides some facts concerning the hierarchy property. Section 3 deals with the gap property. Section 4 introduces the notion of $f$-complete sets whose existence is shown in Section 5, for a variety of time bounds $f$. Theorem 1 says that $\H(f)$ implies a separation of the deterministic from the nondeterministic time classes for those bounds $f$. Finally, Section 6 collects the main results in a figure and discusses relativizations. These are formulated and confirmed in more detail in an appendix.

\section{Basic notions and results around the hierarchy property}

By a {\it time bound}\/, we always mean a function $\fct{f}{\N}{\N}$ such that $n\leq f(n)\leq f(n+1)$ for all $n\in\N$.
For a prefix $\mbox{\rm X}\in\{\mbox{\rm D}, \mbox{\rm N}, \mbox{\rm A}\}$, the {\it deterministic}\/, {\it nondeterministic}\/ and {\it alternating}\/, respectively, {\it time complexity classes}\/ $\XT(f)$ consist of all languages which are accepted by an $\mbox{\rm X}$TM $\frak{M}$ with a time complexity $t_{\frak{M}}\in \Ord(f)$. Under an $\mbox{\rm X}$TM, we here understand a deterministic, nondeterministic and alternating, respectively, {\it Turing machine}\/ with a special read-only input tape and arbitrarily many, say $b\in\N$, additional work tapes. For further details of definitions, the reader is referred to textbooks like \cite{BDG,DK,Re}. By restricting the number of work tapes to some $b\in\N$, i.e., by considering {\it $b$-tape} $\mbox{\rm X}$TMs only, we obtain the complexity classes $\XT_b(f)$. Thus, $\XT(f)=\bigcup_{b\in\Ni}\XT_b(f)$.

For nondeterministic and alternating TMs, there are well-known tape-number reductions:
\bea
(1) \quad \NT(f) & = & \NT_2(f), \en \mbox{ see \cite{BGW}};\\
(2) \quad \AT(f) & = & \AT_1(f), \en \mbox{ see \cite{PPR}}.
\eea
For the deterministic case, one only knows
\bea
(3) \quad \DT(f) & \subseteq & \DT_2(\,f\cdot (\log\circ f)\,), \en \mbox{ see \cite{HS}}.
\eea

Even for the linear time bound $f(n)=n$, the hierarchy conjecture $\H(n)$ has not yet been proved so far.
The hardness of this problem was emphasized by showing that $\H(n)$ does not relativize. We want to recall this result from \cite{He}, since it marks the starting point of the present investigations.
Moreover, we shall come back to relativizations at the end of this paper.

Let $\DT^A(n)$ and $\DT_{b_0}^A(n)$ denote the
complexity classes defined by linear time-bounded deterministic {\it oracle Turing machines}\/ with a separate input tape, arbitrarily many, respectively $b_0$, work tapes and an additional oracle tape on which queries to the oracle set $A$ can be put (and which is erased immediately after any query). Analogously, the classes $\DT^A(f)$ and $\DT_{b_0}^A(f)$ are defined for any time bound $f$, and the {\it relativized hierarchy property}\/ and
{\it relativized gap property}\/, $\H^A(f)$ and $\G^A(f)$, respectively, concern these relativized complexity classes in a straightforward manner. To show the following assertion, as oracle set we used a language $A$ which is $\NS(n)$-complete with respect to linear-time reduction.

\begin{lemma}[\cite{He}, Proposition 16]
$\H(n)$ does not relativize; this means that there is an oracle set $A$ such that $\neg\,\H^A(f)$ holds, i.e.,
 $\DT^A(n)= \DT_{b_0}^A(n)$ for some $b_0\in\N$.
\end{lemma}

Now we are going to present new results.

\begin{lemma}
$\H(n)$ implies $\DT(n)\not=\NT(n)$.
\end{lemma}

\begin{proof}
We employ the following language introduced in \cite{He}:
\bea
V & = & \{ \langle w, \code(\frak{M})^t\rangle: \;
   w\in\{0,1\}^*, |w|\leq t\in\Nplus \mbox{ and }\\
   & & \hspace*{3.0cm}\mbox{$\frak{M}$ is a 2-tape NTM
   that accepts $w$ in $\leq t$ steps}\}.
\eea
Herein $\code(\frak{M})$ is some standard encoding of the NTM $\frak{M}$ as a word over the binary alphabet $\{0,1\}$. Notice that, without loss of generality, we can restrict the complexity classes to languages
 $L\subseteq \{0,1\}^*$.
  $\langle \,\cdot\,,\cdot\,\rangle$ denotes an injective pairing function computable in linear time, with linear-time computable inverses (projections).
It is not hard to show that $V$ is $\NT(n)$-complete with respect to linear-time reductions, see \cite{He}.

Now we show that $\DT(n)=\NT(n)$ implies $\neg\,\H(n)$. If $\DT(n)=\NT(n)$, then $V\in\DT(n)$.
Thus, there would be a $b_0$-tape DTM $\frak{M}_V$ that decides $V$ in linear time.
Moreover, for any language $L\in\DT(n)\cap \{0,1\}^*$, a linear-time reduction of $L$ to $V$ is given by the assignment
\[ w \; \stackrel{\varrho}{\longmapsto} \; \langle w , \code(\frak{M}_L)^{c_L\cdot|w|}\rangle,
\]
where $\frak{M}_L$ is a $2$-tape NTM accepting $L$ within a time complexity $t_{\frak{M}_L}(n)\leq c_L\cdot n$.
The word function  $\varrho$ defined in this way can be computed in linear time by a $1$-tape DTM $\frak{M}_{\varrho}$.
The composition of $\frak{M}_{\varrho}$ with $\frak{M}_V$ yields a $(b_0+1)$-tape DTM deciding $L$ in linear time.
Hence we would have $\DT(n)=\DT_{b_0+1}(n)$, i.e., $\neg\,\H(n)$.
\qed
\end{proof}

The separation $\DT(n)\not=\NT(n)$, also written as DLIN$\not=$NLIN, was shown in \cite{PPST}, see also \cite{BDG,Gu} for some discussion. The proof is based on an involved analysis of computation graphs of deterministic linear-time acceptors.
Since Lemma 2 has easily been shown, a simple proof of $\H(n)$ cannot be expected.

Unfortunately, the proof of Lemma 2 cannot be transferred to superlinear bounds $f$. A reduction analogously to $\varrho$ would then need a time bound $\Ord(f(n))$, and the assumption $V\in\DT(f)$ would only lead to an $\Ord(f\circ f(n))$ decision of any $L\in\DT(f)$ by means of a uniformly bounded number of work tapes. Thus, this subject requires some more effort. We shall come back to it in Section 4.

Here we proceed with showing an upward transfer of $\neg\,\H(n)$.
{\it Time-constructibility}\/ of a time bound $g$ by a $b_g$-tape DTM $\frak{M}_g$ means that $\frak{M}_g$, for any input word $w$,
writes the word $1^{g(|w|)}$ on one of its work tapes and halts after at most $t_{\frak{M}_g}(|w|)\in\Ord(g(|w|))$ steps.

\begin{lemma}
If $\,\DT(f)=\DT_{b_0}(f)$ and the time bound $g$ is time-constructible by a $b_g$-tape DTM, then $\DT(f\circ g)=\DT_{\max(b_0,b_g)+1}(f\circ g)$.
\end{lemma}

\begin{proof}
The lemma is shown by the standard padding technique due to \cite{RF}. So a sketch of the arguments should be sufficient. Let $L$ be accepted by some DTM $\frak{M}$ in time $\Ord(f\circ g(n))$. We put
\[L'= \{w\cdot\$^i : \; i=g(|w|)-|w| \mbox{ and } \frak{M} \mbox{ accepts } w\}. \]
Then $L'\in \DT(f)$, hence it is accepted by a $b_0$-tape DTM $\frak{M}'$ in time $\Ord(f(n))$.
To accept $L$, let a DTM $\frak{M}''$ work as follows: \\
First, the input $w$ is padded to $w'=w\cdot \$^{g(|w|)-|w|}$. This is done in time $\Ord(\,g(|w|)\,)$ by means of $b_g+1$ work tapes, say the first one carries $w'$ finally. Now let this tape correspond to the input tape of $\frak{M}'$ accepting $L'$, and let $\frak{M}''$ simulate $\frak{M}'$ in this way and using $b_0-b_g$ further work tapes if $b_g<b_0$.
Thus, $\frak{M}''$ can be constructed as a $(\max(b_0,b_g)+1)$-tape DTM, and it accepts $L$ in time $\Ord(g(n))+\Ord(f\circ g(n))=\Ord(f\circ g(n))$.
\qed
\end{proof}

\section{Around the gap property}

By a straightforward diagonalization over the DTMs of a fixed tape number, the following separation is obtained, see \cite{Re}:
\bea
(4) \quad \DT_b(f_1) & \subset & \DT_{b+1}(f_2) \en \mbox{ for time bounds $f_1$ and $f_2$ if $f_1\in\ord(f_2)$ and} \\
   & & \hspace*{2.8cm} \mbox{$f_2$ is time-constructible by a $(b+1)$-tape DTM.}
\eea
Borodin's gap theorem \cite{Bo} shows that the condition of time-constructibility of $f_2$ is essential herein.
If we, moreover, assume that $\G(f_1)$ holds and $\DT(f_1)=\DT(f_2)$, fact $(4)$ yields
$\DT_b(f_1)\subset \DT_{b+1}(f_2)\subseteq \DT(f_1)$ for almost all $b\in\N$. So we have
\begin{lemma}
$\G(f)$ implies $\H(f)$ for all time bounds $f$.
\end{lemma}

Also, it holds an upward transfer of $\G(f)$:
\begin{lemma}
If $\G(f)$ and the time bound $g$ is time-constructible, then we have $\G(f\circ g)$.
\end{lemma}
\begin{proof}
From $\DT(f)=\DT(f_2)$ for some time-constructible $f_2$ satisfying $f\in\ord(f_2)$, by standard padding due to \cite{RF}, it follows $\DT(f\circ g)=\DT(f_2\circ g)$, and we have $f\circ g \in\ord(f_2\circ g)$; moreover $f_2\circ g$ is time-constructible.
\qed
\end{proof}

Now we are going to show that $\G(f_1)$ implies a separation $\DT(f_2)\subset \NT(f_2)$, for certain time bounds $f_2$.
This is indirectly proved. Assuming $\DT(f_2)= \NT(f_2)$, the nondeterministic hierarchy theorem can be applied:
\bea
(5) \quad \NT(f_1) & \subset & \NT(f_2) \en \mbox{ for time bounds $f_1$ and $f_2\;$ if $f_1(n+1)\in\ord(f_2(n))$} \\
   & & \hspace*{2.3cm} \mbox{and $f_2$ is time-constructible, \en see \cite{SFM,Za}.}
\eea
Unfortunately, this is weaker than the nondeterministic analogue of $\neg \,\G(f_1)$: for monotonous $f_1$,
from $f_1(n+1)\in\ord(f_2(n))$ it follows that $f_1(n)\in\ord(f_2(n))$, but the converse does not hold in general. To enforce this, we require:
\bea
(+) \quad \mbox{there is a $c\in\Nplus$ such that $f_1(n+1)\leq c \cdot f_1(n)$ for almost all $n\in\N$.}
\eea
This condition is fulfilled by many functions, e.g., by all polynomials or even by $f_1(n)=\lceil n^r\rceil$, where $r$ is a positive rational number, by $f_1(n)=n\cdot (\log(n))^k$, where $k\in\N$, and also by $f_1(n)=2^n$.

From $(+)$ and $f_1(n)\in\ord(f_2(n))$ it follows $f_1(n+1)\in\ord(f_2(n))$, and fact $(5)$ can be applied.
Assuming, moreover, that $\G(f_1)$ and $\DT(f_2)=\NT(f_2)$ for a related bound $f_2$, we would have
\[\DT(f_1)\subseteq\NT(f_1)\subset\NT(f_2)=\DT(f_2)=\DT(f_1),\]
a contradiction. So we have shown
\begin{proposition}
If $\G(f_1)$ for a time bound $f_1$ satisfying $(+)$, then $\DT(f_2)\not=\NT(f_2)$ for any time-constructible bound
$f_2$ with $f_1\in\ord(f_2)$ and $\DT(f_1)=\DT(f_2)$.
\end{proposition}

Thus, $\G(f_1)$ yields separations of determinism from nondeterminism for related bounds $f_2$. Such separations are
 only known for bounds very close to the linear ones, see \cite{PPST,Re}.
  On the other hand, a proof of $\neg\,\G(f_1)$ would also be surprising, since it would mean an essential improvement of the deterministic time hierarchy theorem, at least locally above $f_1$.

\section{Complete sets for subhomogeneous bounds}

So far, complete sets have mainly been considered for complexity classes defined by regular sets of bounds, as for
P, NP, LIN, NLIN etc. The notion of regular set of bounds was introduced in \cite{He}. It ensures a certain robustness of the related complexity classes and useful properties of complete sets of them. The new notion of completeness we are going to introduce now will enable us to generalize Lemma 2 to certain superlinear bounds.

\begin{definition}[$f$-Completeness]
For a time bound $f$, a language $A$ is called $f$-complete if $\,A\in\NT(f)$ and $L\redlin A$ for all $L\in\NT(f)$. Here $\redlin$ means m-reducibility (i.e., many-one reducibility)
via a function computable by a DTM in linear time.
\end{definition}

In order to apply $f$-complete sets, it is useful to require that the bound function fulfils a certain sharpening of
property $(+)$ from the preceding section.
\begin{definition}[Subhomogeneity]
A function $\fct{f}{\N}{\N}$ is said to be subhomogeneous if for any $c\in\Nplus$ there exists a $\overline{c}\in\Nplus$ such that
\[f(c\cdot n) \leq \overline{c}\cdot f(n) \quad \mbox{for all } n\in\Nplus. \]
\end{definition}
Obviously, to ensure this property, for a monotonous (non-decreasing) function $f$ it is sufficient to require that there is some $c'\in\Nplus$ satisfying
\[f(2\cdot n) \leq c'\cdot f(n) \quad \mbox{for all } n\in\Nplus. \]
This condition was already used in \cite{Ka}. There it was also observed that {\it superpolynomial}\/ functions $f$, i.e., such ones with $\lim_{n\to \infty}\frac{f(n)}{n^k} =\infty$ for all $k\in\N$, cannot be subhomogeneous. On the other hand, constant functions and functions like $f(n)=\lceil n^r\rceil$, for positive rationals $r$, or $\log(n)$ and $\log^*(n)$ are obviously subhomogeneous. If non-decreasing functions $f_1$ and $f_2$ are subhomogeneous, then so are $f_1+f_2$, $f_1\cdot f_2$, $f_2\circ f_1$, and $\max(f_1,f_2)$. This is easily proved. For example, if $f_1(c\cdot n) \leq \overline{c}_1\cdot f(n)$ and $f_2(\overline{c}_1\cdot m) \leq \overline{c}_2'\cdot f(m)$ for all $n,m\in\Nplus$,
then
\[(f_2\circ f_1)(c\cdot n)= f_2(f_1(c\cdot n))\leq f_2(\overline{c}_1\cdot f_1(n))
   \leq \overline{c}_2'\cdot f_2(f_1(n)) =
  \overline{c}_2'\cdot(f_2\circ f_1)(n).\]

The following lemma demonstrates the usefulness of the notion of $f$-completeness.
\begin{lemma}
Let $A$ be an $f$-complete language for a subhomogeneous time bound $f$. Then
\bea
\DT(f)=\NT(f) & \mbox{ iff } &  A\in\DT(f), \quad \mbox{ and}\\
\NT(f)=\co\NT(f) & \mbox{ iff } &  A\in\co\NT(f).
\eea
\end{lemma}
\begin{proof}
We show the first assertion, the second one follows analogously. The direction ``$\,\Rightarrow\,$'' is trivial. To prove ``$\,\Leftarrow\,$'', let $A\in\DT(f)$ and $L\in\NT(f)$. If $L\redlin A$ via an m-reduction $\varrho$ and the characteristic function of $A$, $\chi_A$, is deterministically computable in time $\Ord(f)$, then the characteristic function of $L$, $\chi_L=\chi_A\circ\varrho$, can deterministically be computed in time $\Ord(f)$, too.
This holds since $\varrho$ is computable in linear time and $f$ is subhomogeneous: we have
$|\varrho(w)|\leq c\cdot |w|$ for any input word $w$, with a constant $c$, and
$f(|\varrho(w)|)\leq f(c\cdot |w|)\leq \overline{c}\cdot f(|w|)$, with a suitable constant $\overline{c}$.
\qed
\end{proof}

An obvious supplement will be important for the proof of Theorem 1 below:
\begin{lemma}
If, in the above proof of Lemma 6, $\varrho$ is computable by a $1$-tape DTM in linear time and $A\in\DT_{b_0}(f)$, then $\DT(f)=\NT(f)=\DT_{b_0+1}(f)$.
\end{lemma}

\section{Existence of $f$-complete sets}

In order to ensure the existence of $f$-complete sets, a further condition is employed.
\begin{definition}[Superhomogeneity and Quasihomogeneity]
A function $\fct{f}{\N}{\N}$ is said to be superhomogeneous if for any $c\in\Nplus$ there exists a $c'\in \Nplus$ such that
\[ c\cdot f(n) \leq f(c'\cdot n) \quad \mbox{for all } n\in\Nplus.  \]
A function is called {\it quasihomogeneous}\/ if it is both subhomogeneous and superhomogeneous.
\end{definition}

Examples of superhomogeneous functions are the linear functions, the exponential function $f(n)=2^n$, as well as
the functions $f(n)=\lceil n^r\rceil$ for positive rationals $r$. Constant functions and other bounded functions as well as $\log(n)$ are not superhomogeneous. If non-decreasing functions $f_1(n),f_2(n)\geq 1$ are superhomogeneous, then so are $f_1+f_2$, $f_1\cdot f_2$, and $f_2\circ f_1$. If only $f_1$ is superhomogeneous and $f_2$ is non-decreasing, then $f_1\cdot f_2$ is superhomogeneous, too. All these assertions can easily be shown.
For example, if $c\cdot f_1(n) \leq f_1(c'\cdot n)$ with $c'\in\Nplus$, then for any non-decreasing $f_2$ we have
\[c\cdot (f_1\cdot f_2)(n)= c\cdot f_1(n)\cdot f_2(n)\leq f_1(c'\cdot n)\cdot f_2(n)\leq
f_1(c'\cdot n)\cdot f_2(c'\cdot n)=(f_1\cdot f_2)(c'\cdot n).\]

Thus, by the examples and remarks after Definitions 4 and 5, the property of quasihomogeneity applies to a lot of time bounds, among them the polynomials and several others between the linear functions and the polynomials.

\begin{proposition}
If the time bound $f$ is superhomogeneous and time-constructible, then there is an $f$-complete language that also satisfies the supposition of Lemma 7.
\end{proposition}

\begin{proof}
Let $f$ be a time-constructible superhomogeneous time bound. Without loss of generality, we suppose that $f(n)\geq n+1$ for all $n\in\N$. As in the proof of Lemma 2, we restrict the complexity classes to languages $L\subseteq\{0,1\}^*$ and denote the (binary) encoding of an NTM $\frak{M}$ by $\code(\frak{M})$. For an input word $w=x_1x_2\,\ldots\,x_l$, let $\;\overline{w}=\,0x_10x_20\,\ldots\, 0x_l$ and, for any $a\in\N$,
\[ w_{\frak{M}}^a= \:\overline{\code(\frak{M})}\, 1\, w\, 0\, 1^a\,. \]

 An $f$-complete language is defined by
 \bea
 V_f & = & \{ w_{\frak{M}}^a\,:\; w\in\{0,1\}^*,\en \frak{M} \mbox{ is a $2$-tape NTM,} \en a\in\N, \\
  & & \hspace*{1.3cm} |w_{\frak{M}}^a|= c'\cdot |w| \mbox{ for some $c'\in\Nplus$ such that }
   |\code(\frak{M})|\cdot f(|w|) \leq f(c'\cdot |w|), \\
  & & \hspace*{1.3cm} \mbox{and $\,\frak{M}\,$ accepts $w$ in $\leq f(|w|)$ steps}\, \}.
  \eea

 One easily shows that $V_f\in\NT(f)$: Given an input $w'\in\{0,1\}^*$, it is first checked whether $w'=w_{\frak{M}}^a$ for a $w\in\{0,1\}^*$, $a\in\N$ and some $2$-tape NTM $\frak{M}$ such that $|w_{\frak{M}}^a|= c'\cdot |w|$  for some $c'\in\Nplus$.
 Since $f$ is time-constructible and non-decreasing, this check including the computation of $c'$ can be done in time $\Ord(f(|w'|))$. Having $c'$, the condition $\,|\code(\frak{M})|\cdot f(|w|) \leq f(c'\cdot |w|)\,$ is checked in
$\Ord(f(|w'|))$ further steps. Finally, the (nondeterministic) simulation of $f(|w|)$ steps of $\frak{M}$ on the input $w$ is possible in  $|\code(\frak{M})|\cdot f(|w|) \leq f(c'\cdot |w|)=f(|w'|)$ steps.

To show the hardness of $V_f$, let $L\in\NT(f)$ and $L\in\{0,1\}^*$. By the linear speed up and the tape-number reduction of NTMs, there is a $2$-tape NTM $\frak{M}_L$ accepting  $L$ in time $t_{\frak{M}_L}(n)\leq f(n)$.
Since $f$ is superhomogeneous, there is a constant $c'_L\in\Nplus$ such that
 $|\code(\frak{M}_L)|\cdot f(n) \leq f(c'_L\cdot n)$ for all $n\in\Nplus$. Without loss of generality, let $c'_L\geq 2\cdot|\code(\frak{M})|+3$. The word function $\fct{\varrho_L}{\{0,1\}^*}{\{0,1\}^*}$ defined by
 \[ w \; \stackrel{\varrho_L}{\longmapsto} \; w_{\frak{M}}^a\, , \en \mbox{ where } a= c'_L\cdot|w| - 2\cdot |\code(\frak{M}_L)|-|w| -2\, , \]
is an m-reduction of $L$ to $V_f$, and it is computable by a $1$-tape DTM in linear time.
\qed
\end{proof}

Now we are well prepared to state our main result.

\begin{theorem}
If the time bound $f$ is time-constructible and quasihomogeneous, then $\H(f)$ implies the separation
$\DT(f)\not= \NT(f)$.
\end{theorem}
\begin{proof}
We conclude indirectly: from $\DT(f)= \NT(f)$, by Proposition 2 and Lemma 7, it would follow that $\DT(f)=\DT_{b_0+1}(f)$ for some $b_0\in\N$, hence $\neg\,\H(f)$.
\qed
\end{proof}

It might be of interest to note that the proof of Proposition 2 can be modified to obtain languages that are complete in space complexity classes, with respect to linear-time reduction. More precisely, employing the usual notations, we have
\begin{corollary}
If $f$ is a superhomogeneous and time-constructible time bound, there is a language
$V\!S_f\in \mbox{\rm NSPACE}(f)$ such that $L\redlin V\!S_f$ for all $L\in\mbox{\rm NSPACE}(f)$.
\end{corollary}
\begin{proof}
We put
 \bea
 V\!S_f & = & \{ w_{\frak{M}}^a\,:\; w\in\{0,1\}^*,\en \frak{M} \mbox{ is a single-tape NTM,} \en a\in\N, \\
  & & \hspace*{1.3cm} |w_{\frak{M}}^a|= c'\cdot |w| \mbox{ for some $c'\in\Nplus$ such that }
   |\code(\frak{M})|\cdot f(|w|) \leq f(c'\cdot |w|), \\
  & & \hspace*{1.3cm} \mbox{and $\,\frak{M}\,$ accepts $w$ on $\leq f(|w|)$ cells}\, \}.
  \eea
Herein, by a {\it single-tape TM}\/, we mean a usual TM with only one tape, on which the input is given and the work has to be carried out. The remaining part is analogous to the proof of Proposition 2.
\qed
\end{proof}

If, moreover, $f$ is even quasihomogeneous, the deterministic and nondeterministic linear time classes relativized to  $V\!S_f$ coincide. More precisely, we have
 LIN$^{\mbox{\footnotesize $V\!S_f$}}=\,$NLIN$^{\mbox{\footnotesize $V\!S_f$}}=\,$NSPACE$(f)$. The proof is quite analogous to that of Proposition 15 in \cite{He}.

\section{Conclusion and relativizations}

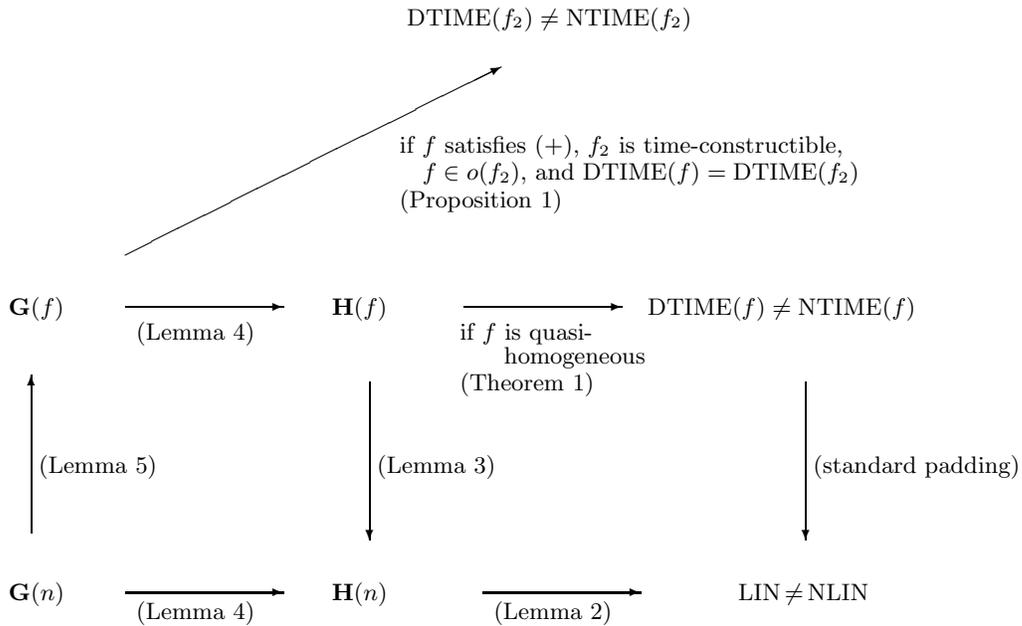
\begin{figure}
\unitlength1cm
\begin{picture}(16,8.2)
\put(0.75,0.4){\begin{picture}(13,7.7)
 \put(0.2,0.2){$\G(n)$}
 \put(4.5,0.2){$\H(n)$}
 \put(9.9,0.2){LIN$\,\not=\,$NLIN}
 \put(0.2,4){$\G(f)$}
 \put(4.5,4){$\H(f)$}
 \put(8.7,4){$\DT(f)\not=\NT(f)$}
 \put(6.2,3.675){$\mbox{\small if $f$ is quasi-}$}
 \put(6.79,3.35){$\mbox{\small homogeneous}$}
 \put(6.2,3){$\mbox{\small (Theorem 1)}$}
 \put(5.5,7.85){$\DT(f_2)\not=\NT(f_2)$}
 \put(5.4,6.15){$\mbox{\small if $f$ satisfies $(+)$, $f_2$ is time-constructible,}$} %
 \put(5.7,5.8){$\mbox{\small $f\in o(f_2)$, and $\DT(f)=\DT(f_2)$ }$}
 \put(5.4,5.425){$\mbox{\small (Proposition 1) }$}

 \put(1.75,0.3){\vector(1,0){2.1}}
 \put(1.9,-0.05){$\mbox{\small (Lemma 4)}$}
 \put(6.5,0.3){\vector(1,0){2.1}}
 \put(6.65,-0.05){$\mbox{\small (Lemma 2)}$}
 \put(1.75,4.1){\vector(1,0){2.1}}
 \put(1.9,3.675){$\mbox{\small (Lemma 4)}$}
 \put(6.25,4.1){\vector(1,0){2.1}}
 \put(1.75,4.8){\vector(2,1){5}}  %
 \put(0.5,1.1){\vector(0,1){2.1}}
 \put(0.6,1.9){$\mbox{\small (Lemma 5)}$}
 \put(5,3.1){\vector(0,-1){2.1}}
 \put(5.1,1.9){$\mbox{\small (Lemma 3)}$}
 \put(10.8,3.1){\vector(0,-1){2.1}}
 \put(10.9,1.9){$\mbox{\small (standard padding)}$}
\end{picture}}
\end{picture}
\normalsize
\caption{ Main results (for time-constructible bounds $f\,$)}
\end{figure}

Figure 1 summarizes the main implications we have shown. Herein let $f$ be any time-constructible time bound.
The labels at the arrows hint to the corresponding results within this paper and recall some necessary suppositions. The implication indicated by the rightmost vertical arrow, that $\DT(f)\not=\NT(f)$ implies LIN$\,\not=\,$NLIN,
is well-know. It follows by standard padding due to \cite{RF}: LIN$=$ NLIN would yield $\DT(f)=\NT(f)$. Notice that LIN$\,\not=\,$NLIN is the only claim occurring in the figure that was already proved, see \cite{PPST}.

\begin{proposition}
All the implications shown in Figure 1 relativize, i.e., they remain valid in their relativized form with respect to an arbitrary ocacle $A$.
\end{proposition}

For a discussion of the role of relativization, we refer to \cite{Fo}.
The proof of Proposition 3 is rather straightforward. So the details will be given in an appendix. There we re-formulate the main results in their relativized versions and deal with the related adaptations of proof arguments.

  Thus, for any oracle $A$ satisfying LIN$^A=$NLIN$^A=\DT^A(n)$, as it was given by Proposition 15 in \cite{He}, it follows that both $\neg\,\H^A(n)$ and $\neg\,\G^A(n)$ hold, as well as $\neg\,\H^A(f)$ and $\neg\,\G^A(f)$ for any time-constructible bound $f$. Hence $\G(f)$ and $\H(f)$ do not relativize if they hold.

Since $\G(n)$ has strong consequences, e.g., $\H(f)$ and $\DT(f)\not=\NT(f)$ for all quasihomogeneous time-constructible $f$, a proof of $\G(n)$ seems to lie beyond the present abilities.
Unfortunately, we even have to leave open how to construct an oracle $B$ satisfying $\G^B(n)$. By Proposition 3,  $\G^B(n)$ implies $\H^B(n)$ but also $\G^B(f)$ and $\H^B(f)$ for any time-constructible time bound $f$.
For any computable time bound $f(n)\in\ord(2^n)$, a computable oracle $B$ with $\DT^B(f)\not=\NT^B(f)$ can be obtained by a straightforward adaptation of the standard construction of $B$ satisfying $\P^B\not=\NP^B$ due to \cite{BGS}.

The results of this paper should have demonstrated the usefulness of the notions introduced in Definitions 1-5.
It turned out that hierarchy property and gap property are related not only to each other but also to the central problem of separation of determinism from nondeterminism. This stresses the hardness of $\H(f)$ and $\G(f)$.
As a particular point, we still remark that the existence of $f$-complete sets has been used only in a rather special way in this paper. It would be interesting to look for further substantial applications of this concept.

\section*{Acknowledgement}
I am grateful to readers and referees of earlier versions of this paper whose comments and hints contributed to improvements of the presentation.

\section*{Appendix. Details of the relativizations}

This appendix is devoted to the relativized versions of the main results of the paper. They will precisely be formulated in order to confirm Proposition 3. As far as the proofs are only straightforward translations of those of the unrelativized versions, they will merely be sketched. Sometimes, however, one has to be more careful, since, e.g., unrelativized reductions have to be applied to (members of) relativized complexity classes. We shall emphasize these details.

Throughout this section, let $A$ be a fixed language taken as oracle set.

\begin{lemma}[Lemma 2 relativized]
$\H^A(n)$ implies $\DT^A(n)\not=\NT^A(n)$.
\end{lemma}

\begin{proof}
We employ the relativized version of $V$ which was already used in \cite{He}:
\bea
V^A & = & \{ \langle w, \code(\frak{M})^t\rangle: \;
   w\in\{0,1\}^*, |w|\leq t\in\Nplus \mbox{ and $\frak{M}$ is a 2-tape oracle NTM} \\
   & & \hspace*{3.0cm}\mbox{such that $\frak{M}^A$ accepts $w$ in $\leq t$ steps}\}.
\eea
Herein let $\code(\frak{M})$ be a standard encoding of the oracle NTM $\frak{M}$ as a word over the binary alphabet $\{0,1\}$. It does not depend on $A$. $\frak{M}^A$ denotes the machine $\frak{M}$ working with oracle $A$.
Again, we restrict the complexity classes to languages  $L\subseteq \{0,1\}^*$.
Moreover, it has to be noticed that fact (1) relativizes. It can easily be proved that
$V^A$ is $\NT^A(n)$-complete with respect to linear-time reductions, see \cite{He}.

We proceed with showing that $\DT^A(n)=\NT^A(n)$ yields $\neg\,\H^A(n)$. If $\DT^A(n)=\NT^A(n)$, then $V^A\in\DT^A(n)$.
There would be a $b_0$-tape oracle DTM $\frak{M}_V$ such that $\frak{M}_V^A$ decides $V^A$ in linear time.
For any language $L\in\DT(n)^A\cap \{0,1\}^*$, an (unrelativized) linear-time reduction of $L$ to $V^A$ is given by the assignment
\[ w \; \stackrel{\varrho}{\longmapsto} \; \langle w , \code(\frak{M}_L)^{c_L\cdot|w|}\rangle,
\]
where $\frak{M}_L$ is a $2$-tape oracle NTM such that $\frak{M}_L^A$ accepts $L$ in linear time.
The word function  $\varrho$ defined in this way can be computed in linear time by a $1$-tape DTM $\frak{M}_{\varrho}$.
The composition of $\frak{M}_{\varrho}$ with $\frak{M}_V$ yields a $(b_0+1)$-tape oracle DTM deciding $L$ in linear time with respect to the oracle $A$.
Hence $\DT^A(n)=\DT_{b_0+1}^A(n)$, i.e., $\neg\,\H^A(n)$.
\qed
\end{proof}

\begin{lemma}[Lemma 3 relativized]
If $\,\DT^A(f)=\DT^A_{b_0}(f)$ and the time bound $g$ is time-constructible by a $b_g$-tape DTM, then $\DT^A(f\circ g)=\DT^A_{\max(b_0,b_g)+1}(f\circ g)$.
\end{lemma}

\begin{proof}
It is folklore that the technique of standard padding relativizes. Indeed, the proof of Lemma 3 can straightforwardly be relativized, where a $b_g$-tape DTM computing $g$ in time $\Ord(g)$ has to be combined with a suitable
$b_0$-tape oracle DTA.
The details are left to the reader.
\qed
\end{proof}

The standard diagonalization proving fact (4) relativizes, too. Hence our proof of Lemma 4 can immediately be translated into its relativized version, and we have
\begin{lemma}[Lemma 4 relativized]
$\G^A(f)$ implies $\H^A(f)$ for all time bounds $f$.
\end{lemma}

Since Lemma 5 follows by standard padding, we immediately have
\begin{lemma}[Lemma 5 relativized]
If $\G^A(f)$ and the time bound $g$ is time-constructible, then it holds $\G^A(f\circ g)$.
\end{lemma}

To show the relativized version of Proposition 1, it is enough to remark that fact (5), the nondeterministic hierarchy theorem, relativizes. This indeed holds, see the proof in \cite{Za}. So it straightforwardly follows
\begin{proposition}[Proposition 1 relativized]
If $\G^A(f_1)$ for a time bound $f_1$ satisfying $(+)$, then $\DT^A(f_2)\not=\NT^A(f_2)$ for any time-constructible bound $f_2$ with $f_1\in\ord(f_2)$ and $\DT^A(f_1)=\DT^A(f_2)$.
\end{proposition}

Now we are going to relativize a combination of Lemmas 6 and 7 which is just needed to prove the relativized version of Theorem 1. First the notion of $f$-completeness is modified accordingly.
A language $B$ is said to be {\it $(f,A,1)$-complete}\/ if $\,B\in\NT^A(f)$ and $L\redlino B$ for all $L\in\NT^A(f)$. Here $\redlino$ means m-reducibility
via a function computable by a $1$-tape DTM in linear time.

\begin{lemma}[from Lemmas 6, 7 relativized]
Let $B$ be an $(f,A,1)$-complete language for a subhomogeneous time bound $f$.
If $B\in\DT^A_{b_0}(f)$ for some $b_0\in \Nplus$,
then $\DT^A(f)=\NT^A(f)=\DT^A_{b_0+1}(f)$.
\end{lemma}

\begin{proof} The nontrivial part  of the proof of Lemma 6 can word-by-word be translated into its relativized version.
\qed
\end{proof}

\begin{proposition}[Proposition 2 relativized]
If the time bound $f$ is superhomogeneous and time-constructible, then there is an $(f,A,1)$-complete language.
\end{proposition}

\begin{proof}
Applying the notations from the proof of Proposition 2, we put
 \bea
 V^A_f & = & \{ w_{\frak{M}}^a\,:\; w\in\{0,1\}^*,\en \frak{M} \mbox{ is a $2$-tape oracle NTM,} \en a\in\N, \\
  & & \hspace*{1.3cm} |w_{\frak{M}}^a|= c'\cdot |w| \mbox{ for some $c'\in\Nplus$ such that }
   |\code(\frak{M})|\cdot f(|w|) \leq f(c'\cdot |w|), \\
  & & \hspace*{1.3cm} \mbox{and $\,\frak{M}^A\,$ accepts $w$ in $\,\leq f(|w|)$ steps}\, \}.
  \eea
 Both that $V^A_f\in\NT^A(f)$ as well as the hardness of $V^A_f$ for the class $\NT^A(f)$ with respect to $\redlino$ can be shown by straightforward modifications of the related parts of the proof of Proposition 2.
\qed
\end{proof}

Finally, by adapting the short proof of Theorem 1, from Proposition 5 and Lemma 12 it follows
\begin{theorem}[Theorem 1 relativized]
If the time bound $f$ is time-constructible and quasihomogeneous, then $\H^A(f)$ implies the separation
$\DT^A(f)\not= \NT^A(f)$.
\end{theorem}


\begin{thebibliography}{999}
\bibitem{Aa} S.O.~Aanderaa, On $k$-tape versus $(k-1)$-tape real time computation, in: R.M.~Karp (Ed.),
Complexity of Computation, SIAM-AMS Proceedings 7, Providence, Rhode Island, 1974, pp.~75-96.
\bibitem{BGS} T.~Baker, J.~Gill, and R.~Solovay, Relativizations of the $\P=?\,\NP$ question,
SIAM Journal on Computing 4 (1975) 431-442.
\bibitem{BDG} J.L.~Balc\'{a}zar, J.~D\'{\i}az, and J.~Gabarr\'{o},
Structural Complexity II, Springer, Berlin, 1990.
\bibitem{BGW} R.V.~Book, S.A.~Greibach, and B.~Wegbreit, Time and tape bounded Turing acceptors and AFL's,
Journal of Computer and System Sciences 4 (1970) 606-621.
\bibitem{Bo} A.~Borodin, Computational complexity and the existence of complexity gaps, Journal of the ACM 19 (1972) 158-174.
\bibitem{CN} S.~Cook and P.~Nguyen, Logical Foundations of Proof Complexity,
Cambridge University Press, New York, 2010.
\bibitem{DK} D.-Z.~Du and K.-I Ko, Theory of Computational Complexity, Wiley-Interscience, New York, 2000.
\bibitem{Fo} L.~Fortnow, The role of relativization in complexity theory,
Bulletin of the EATCS 52 (1994) 229-244.
\bibitem{Gu} S.~Gupta, Alternating time versus deterministic time: a separation,
Mathematical Systems Theory 29 (1996) 661-672.
\bibitem{He} A.~Hemmerling, Observations on complete sets between linear time and polynomial time,
Information and Computation 209 (2011) 173-182.
\bibitem{Hen} F.C.~Hennie, One-tape, off-line Turing machine computations, Information and Control 8 (1965) 553-578.
\bibitem{HS} F.C.~Hennie and R.E.~Stearns, Two-tape simulation of multitape Turing machines,
Journal of the ACM 13 (1966) 533-546.
\bibitem{Ka} R.~Kannan, Towards separating nondeterminism from determinism,
Mathematical Systems Theory 17 (1984) 29-45.
\bibitem{MSST} W.~Maass, G.~Schnitger, E.~Szemer\'{e}di, and G.~Tur\'{a}n, Two tapes versus one for off-line Turing machines,  Computational Complexity 3 (1993) 392-401.
\bibitem{PSSN} R.~Paturi, J.I.~Seiferas, J.~Simon, and R.E.~Newman-Wolfe, Milking the Aanderaa argument, Information and Computation 88 (1990) 88-104.
\bibitem{Pa} W.~Paul, On-line simulation of $k+1$ tapes by $k$ tapes requires nonlinear time, Information and Control 53 (1982) 1-8.
\bibitem{PPR} W.J.~Paul, E.J.~Prauss, and R.~Reischuk, On alternation,
Acta Informatica 14 (1980) 243-255.
\bibitem{PPST} W.~Paul, N.~Pippenger, E.~Szemer\`{e}di, and W.T.~Trotter,
On determinism versus nondeterminism and related problems,
Proceedings of the 24th IEEE Symposium on Foundations of Computer Science (1983) 429-438.
\bibitem{Re} K.~R.~Reischuk, Einf\"{u}hrung in die Komplexit\"{a}tstheorie, B.~G.~Teubner, Stuttgart, 1990.
\bibitem{RF} S.~Ruby and P.C.~Fischer, Translational methods and computational complexity,
Proceedings of the 6th Annual IEEE Symposium on Switching Circuit Theory and Logical Design (1965) 173-178.
\bibitem{SFM} J.I.~Seiferas, M.J.~Fischer, and A.R.~Meyer, Separating nondeterministic time complexity classes,
Journal of the ACM 25 (1978) 146-167.
\bibitem{TYL} K.~Tadaki, T.~Yamakami, and J.C.H.~Lin, Theory of one-tape linear-time Turing machines,
Theoretical Computer Science 411 (2010) 22-43.
\bibitem{Za} S.~\v{Z}\`{a}k, A Turing machine time hierarchy, Theoretical Computer Science 26 (1983) 327-333.

\end{thebibliography}
\end{document}